\documentclass[11pt,a4paper]{article}
\usepackage[latin1]{inputenc}
\usepackage{authblk}
\usepackage{amsfonts}
\usepackage{amscd}
\usepackage{amsmath}
\usepackage{theorem}
\usepackage{mathrsfs}
\usepackage{fancybox,amssymb}
\usepackage{geometry}
\usepackage[english]{babel}
\usepackage{graphicx}
\usepackage{epstopdf}

\newcommand{\mF}{\mathcal{F}}

\newcommand{\R}{\mathbb{R}}

\newcommand{\mP}{\mathbb{P}}
\newcommand{\mE}{\mathbb{E}}
\newcommand{\md}{\,{\rm d}}

\newcommand{\s}{\sum\limits}

\newcommand{\w}{\wedge}

\newcommand{\one}{1\mkern-5mu{\hbox{\rm I}}}

\theoremstyle{break}
\theorembodyfont{}
\newtheorem{Def}{Definition}[section]
\newtheorem{Bem}[Def]{Remark}

\newtheorem{Lem}[Def]{Lemma}
\newtheorem{Satz}[Def]{Proposition}

\newtheorem{Thm}[Def]{Theorem}
\newtheorem{Bsp}[Def]{Example}

\makeatletter
\newenvironment{proof}{\noindent{\textit{Proof:}}}{%
\unskip\nobreak\hfil\penalty50\hskip1em\null\nobreak
$\Box$
\parfillskip=\z@\finalhyphendemerits=0\endgraf\bigskip}

\let\oldendBsp\endBsp
\def\endBsp{\unskip\nobreak\hfil\penalty50\hskip1em\null\nobreak\hfil%
$\blacksquare$\parfillskip=\z@\finalhyphendemerits=0\endgraf\oldendBsp}
\let\oldendBem\endBem
\def\endBem{\unskip\nobreak\hfil\penalty50\hskip1em\null\nobreak\hfil%
$\blacksquare$\parfillskip=\z@\finalhyphendemerits=0\endgraf\oldendBem}
\makeatother

\date{}

\title{An Exponential Cox--Ingersoll--Ross Process as Discounting Factor}

\author[1]{Julia Eisenberg \thanks{jeisenbe@tuwien.ac.at}}
\author[2]{Yuliya Mishura \thanks{myus@univ.kiev.ua}}
\affil[1]{\small TU Wien/University of Liverpool}
\affil[2]{Taras Shevchenko National University of Kyiv}
\begin{document} 
\maketitle
\begin{abstract}\noindent
We consider an economic agent (a household or an insurance company) modelling its surplus process by a deterministic process or by a Brownian motion with drift. The goal is to maximise the expected discounted spendings/dividend payments, given that the discounting factor is given by an exponential CIR process.
\\In the deterministic case, we are able to find explicit expressions for the optimal strategy and the value function.
\\For the Brownian motion case, we offer a method allowing to show that for a small volatility the optimal strategy is a constant-barrier strategy.
\vspace{6pt}
\noindent
\\{\bf Key words:} Hamilton--Jacobi--Bellman equation, Cox--Ingersoll--Ross process, dividends,  Brownian risk model, consumption.
\settowidth\labelwidth{{\it 2010 Mathematical Subject Classification: }}%
                \par\noindent {\it 2010 Mathematical Subject Classification: }%
                \rlap{Primary}\phantom{Secondary}
                93E20\newline\null\hskip\labelwidth
                Secondary 91B30, 60K10
\end{abstract}
\section{Introduction}
\subsection{General Introduction}
An insurance company's credit rating indicates its ability to pay customer's claims. A bad credit rating can affect company's business plan, growth potential or even survival chances if new finance is needed to fulfil the capital requirements prescribed by Solvency II. 
The rating process run by a credit rating agency includes quantitative and qualitative analysis, where cash flow is one of the most important factors. A particular attention is paid to dividend payments, which are commonly believed to indicate company's financial health. 
Searching for the optimal strategy maximising the value of expected discounted dividends under different constraints and in different setups has been a popular problem in actuarial mathematics for a long time. The papers by Shreve et al. \cite{shreve}, Asmussen and Taksar \cite{astak}, Azcue and Muler \cite{azcue} are just some examples. For a detailed review we refer for instance to the survey by Albrecher and Thonhauser \cite{albthreview}. The papers mentioned above assume the discounting rate to remain constant up to the considered time horizon, often chosen to be infinite.
Following the recent crisis with ultra low interest rates in Europe, the question arises whether the discounting of cash flows by a constant discounting rate could be considered as an admissible assumption. A stochastic discounting factor increases the dimension of the considered problem along with the complexity. Nevertheless, in the recent years stochastic discounting has become a topical question inter alia in dividend maximisation problems.  
For instance in \cite{jp}, the interest rate is modelled by a positive deterministic function of the current state of a given Markov chain. If the drift of the underlying surplus process is positive in each state, Jiang and Pistorius \cite{jp} prove that it is optimal to adopt a regime-dependent barrier strategy; if the drift is small and negative in one state, the optimal strategy has a different form, which is explicitly identified for two regimes case.
\\Akyildirim et al.\ consider in \cite{agrs} two macroeconomic factors: the interest rates and the issuance costs. Both factors are assumed to be governed by an exogenous Markov chain. The optimal dividend policy is characterised in dependence on these two factors: all things being equal, firms distribute more dividends when interest rates are high and less when issuing costs are high. 
\\Whereas Jiang and Pistorius \cite{jp} use the fixed point theorem in order to obtain their results, Akyildirim et al.\ \cite{agrs} apply the direct approach by solving the corresponding ODEs, a method we will use in our paper.  

In the present paper, we are taking into account the time-varying interest by introducing a discounting factor given by an exponential Cox--Ingersoll--Ross (CIR) process. A CIR process is a squared diffusion process, which can attain non-negative values and hit zero for special parameters. 
Usually, by modelling interest rates one assumes CIR to be mean-reverting. Under this assumption, our problem would be ill-posed. Therefore, we require CIR process to be non-mean-reverting, implying the almost sure convergence to infinity. 

We assume that the underlying income process is a linear function of time without a random component. Our target is to maximise the expected discounted consumption. This structure yields a two-dimensional problem where the optimal consumption strategy depends on the parameters of the underlying CIR process. For instance, for a highly volatile discounting factor, it might be optimal to wait with the consumption until the discounting process approaches some relative small positive level, taking into account that the waiting period could last forever. In the low volatility case, we prove that the optimal strategy will be to always spend the maximal possible amount independent of the discounting factor.
\\As an example, we consider an insurance company whose surplus is described by a Brownian motion with drift independent of the CIR. Here, we again have a two-dimensional problem. However, the problem formulation puts an emphasis on the ruin time of the underlying surplus process. We are able to reduce the problem to the classical setup with a constant discounting rate for some special parameters of the CIR process.

To the best of our knowledge, this paper is the first to study an exponential CIR as a discounting factor in the context of consumption/dividend maximisation problems. Despite the fact that the value function depends on two variables -- the surplus and the discounting process -- we are able to find explicit expressions for the optimal strategy and the value function in the deterministic income case and (under some restrictions on the underlying CIR) in the case of Brownian risk model. 

It will be of major importance for the understanding of the paper to remind the reader on some properties and results connected to CIR processes.
Accordingly, we organised the paper as follows: in the next subsection we give an overview over CIR processes. For the convenience of reading, we postpone the technical proofs to the appendix.
\\In Section 2, we consider the case of a deterministic, linear in time income process, which can be interpreted as the income of an individual or household. There, we will distinguish between two different cases concerning the parameters of the considered CIR process and give explicit expressions for the optimal strategy and the value function. Here, we solve the problem of dividend maximisation for special parameters of the underlying CIR process. Conclusion at the end of Sections 2 gives an overview over the possible future research directions. Some technical proofs are given in the appendix, Section \ref{app}. 
\subsection{Preliminaries}\label{facts:cir}
\noindent
For the sake of clarity of presentation, we postpone the most proofs of this subsection to the appendix, Section \ref{app}.
\begin{figure}
\includegraphics[scale=0.2, bb = -100 -100 0 600]{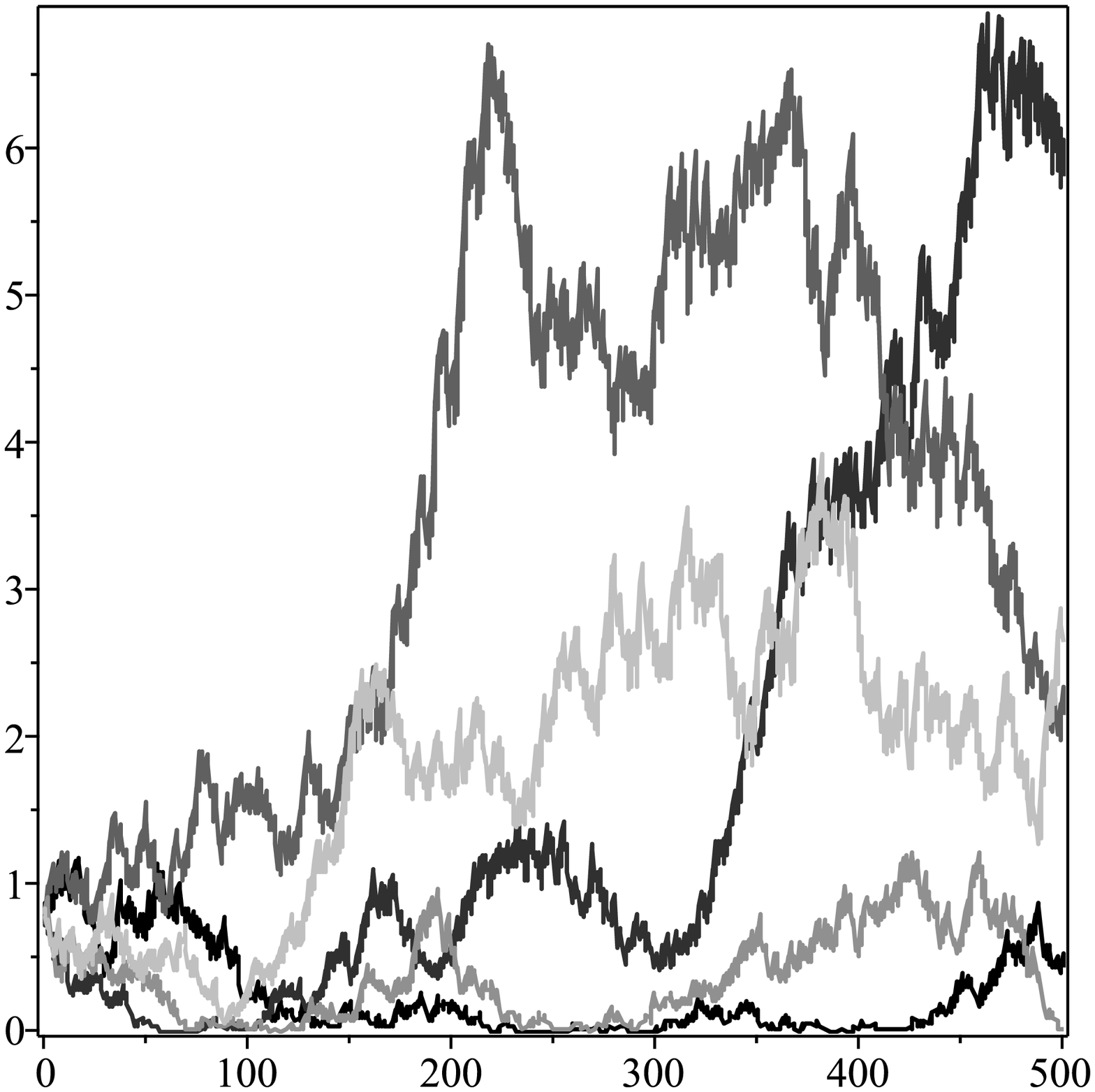}
\includegraphics[scale=0.2, bb = -600 -100 -500 600]{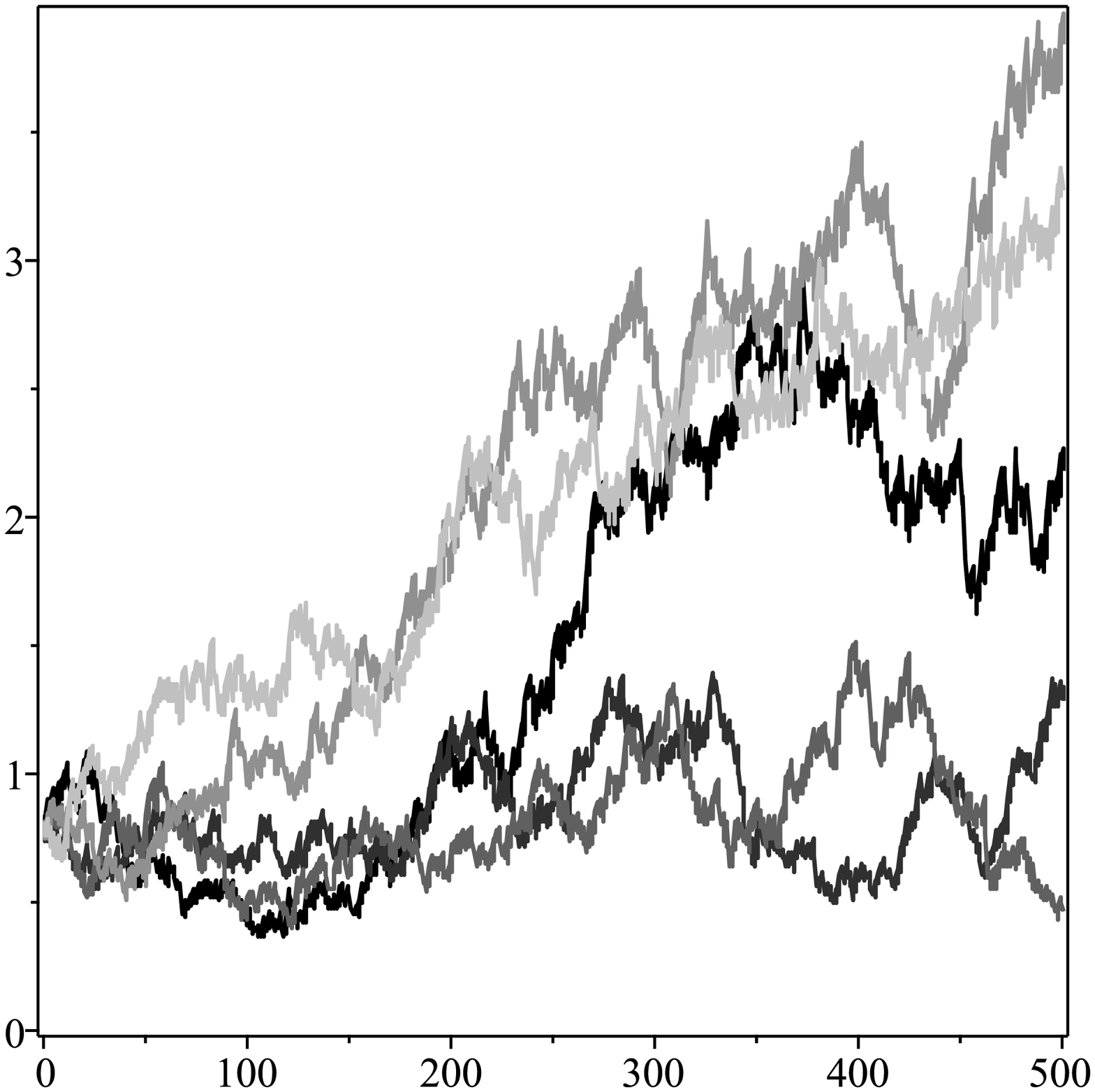}
\includegraphics[scale=0.2, bb = -1100 -100 -1000 600]{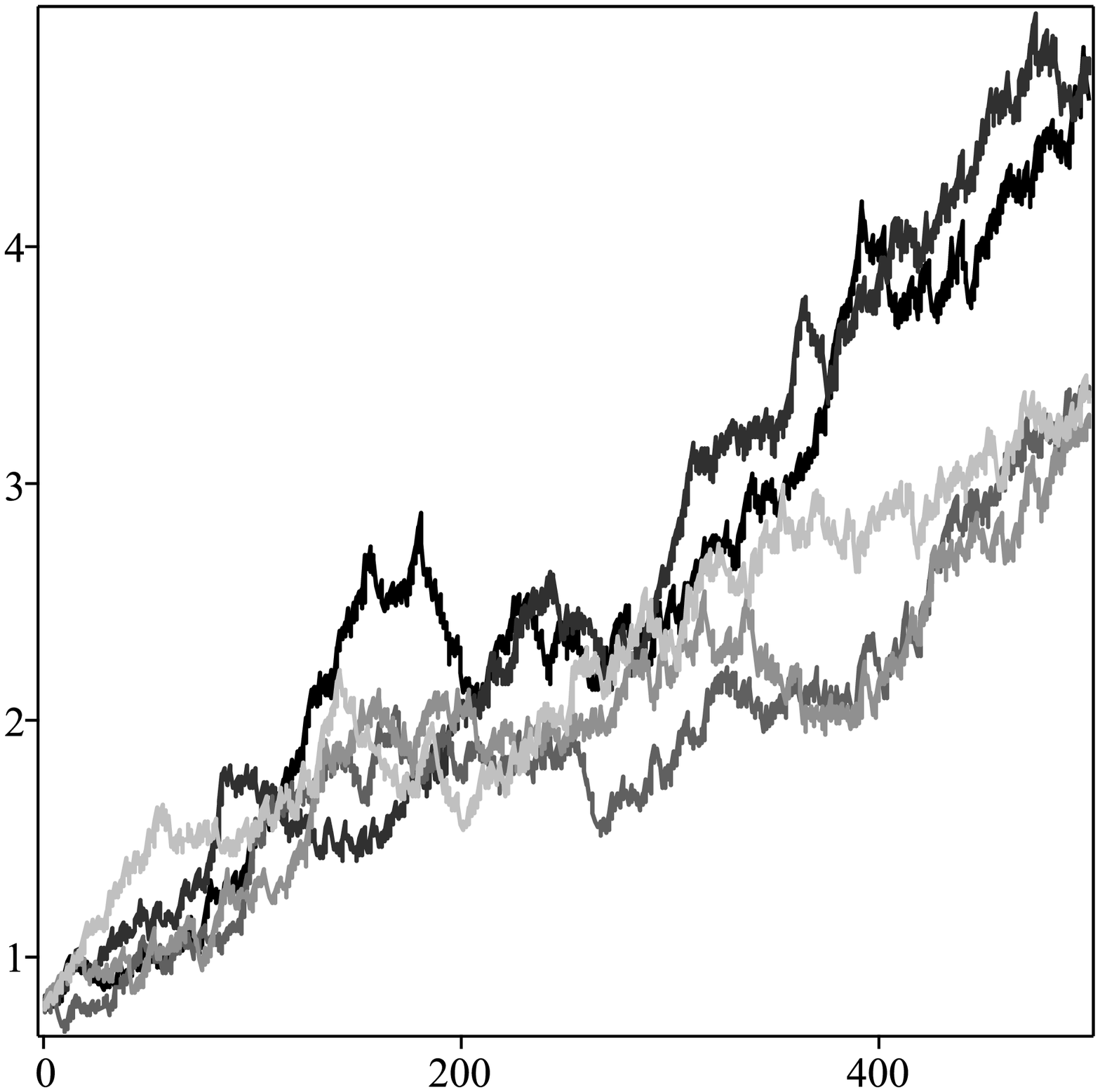}
\caption{Paths of a CIR process with $a=0.001$, $b=0.002$, $\delta_1=0.09$ (left picture), $\delta_2=0.045$ (middle) and $\delta_3=0.02$ (right).\label{fig1}}
\end{figure}
Here and in the following we use the common notation: 
$
\mP[\;\cdot\;|Y_0=y]=\mP_y[\;\cdot\;]$ and $\mE[\;\cdot\; |Y_0=y]=\mE_y[\;\cdot\;]$ for any stochastic process $\{Y_t\}$.
In the remainder of the paper we let $r=\{r_t\}$ be a Cox--Ingersoll--Ross (CIR) process 
\begin{equation}
\md r_t=(a r_t+b)\md t+\delta\sqrt{r_t}\md W_t,\label{intro:cir}
\end{equation}
where $a$, $b$ and $\delta$ are positive constants and $W=\{W_t\}$ is a standard Brownian motion. Due, for example, to \cite{ek}, CIR processes have the strong Markov property.
We define
\[
M(r,t):=\mE_r[e^{-r_t}]\;.
\]
Due to \cite{cir}, we know that the density function of $r_t$ with initial value $r$ is given by
\begin{align}
f(y):=c(t)e^{-u(t,r)-v(t,y)}\Big(\frac{v(t,y)}{u(t,r)}\Big)^{q/2}I_q\big(2\sqrt{u(t,r)v(t,y)}\big),\label{density}
\end{align}
where $I_q(x)=\s_{m=0}^\infty \frac 1{m!\Gamma(m+q+1)}\Big(\frac{x}2\Big)^{2m+q}$ is the modified Bessel function of the first kind and
\begin{align*}
\quad &c(t):=\frac{2a}{(e^{at}-1)\delta^2}, &&q:=\frac {2b}{\delta^2}-1,
\\&u(t,r):=c(t)re^{at}, && v(t,y):=c(t)y.
\end{align*}
Also, one has
\begin{align}
&M(r,t)=\mE_r[e^{- r_t}]= e^{-\frac {2ab}{\delta^2}t}\beta(t)^{\frac{2b}{\delta^2}}\cdot e^{-r \beta(t)},\label{mef}
\end{align}
where $\beta(t):=\frac{1}{\frac{\delta^2}{2a}+\big(1-\frac{\delta^2}{2a}\big)e^{-at}}$.
\begin{Lem}\label{intro:convex}
In the case $\frac{\delta^2}{2}\le a$, the function $M(r,t)$ is strictly decreasing in $t$ and the process $\{e^{-r_t}\}$ is a supermartingale.
\end{Lem}
\begin{proof}
Using that $\beta'(t)=a\big(1-\frac{\delta^2}{2a}\big)e^{-at}\beta(t)^2>0$, we obtain
\begin{align*}
M_t(r,t)=M(r,t)\Big\{-b \beta(t)-r ae^{-at}\big(1-\frac{\delta^2}{2a}\big)\beta(t)^2\Big\}<0
\end{align*}
for all $r\in\R_+$. The supermartingale property follows immediately due to the Markov property and the structure of $M$. 
\end{proof}
\begin{Lem}\label{inro:diffeq}
Due to \cite[p.\ 282]{revuzyor}, the function $M$ solves the partial differential equation
\[
(ar+b)M_r(r,t)+\frac{\delta^2 r}{2}M_{rr}(r,t)-M_t(r,t)=0\;.
\]
\end{Lem}
The below lemma ensures the well-definiteness of the problems we are going to consider.
\begin{Lem}\label{app:lem1}
If $a>0$ then the CIR process $\{r_t\}$ fulfils $\lim\limits_{t\to\infty}r_t=\infty$ a.s. 
%\hfill $\blacksquare$
\end{Lem}
For the proof confer the appendix, Section \ref{app}.\smallskip
\\The usual method for proving a verification theorem is to apply Ito's formula, and to prove the stochastic integral to be a martingale. Later, we will see that the following result provides the necessary martingale argument for the verification theorem.  
\begin{Lem}\label{lem:endl}
For $2b<\delta^2$ let $q$ be given like in \eqref{density} then it holds
\begin{align*}
&\int_0^\infty y^{-\frac{2b}{\delta^2}} e^{-\frac{2a}{\delta^2}y}\md y<\infty,
\\&\int_0^s \mE_r\big[r_t^{-2q-1}\big]\md t<\infty \quad\mbox{for all $s\in\R_+$}.
\end{align*}
\end{Lem}
\begin{proof}
Due to $2b<\delta^2$ it holds $-1<q<0$ and 
\begin{align*}
\int_0^\infty y^{-\frac{2b}{\delta^2}} e^{-\frac{2a}{\delta^2}y}\md y&= \int_0^\infty y^{\big(1-\frac{2b}{\delta^2}\big)-1} e^{-\frac{2a}{\delta^2}y}\md y
=\Gamma(-q)\Big(\frac{2a}{\delta^2}\Big)^{q}<\infty\;.
\end{align*}
Further, using \eqref{density} and the bounded convergence theorem, we get:
\begin{align*}
\mE_r\big[r_t^{-2q-1}\big]&=\s_{m=0}^\infty\frac{c(t)^{q+1+2m}e^{-c(t)re^{at}}}{m!\Gamma(m+q+1)}\int_0^\infty y^{m-\frac{2b}{\delta^2}}e^{-\frac{4a}{\delta^2}y}e^{-c(t)y}\md y
\\&=\s_{m=0}^\infty\frac{c(t)^{q+1+2m}e^{-c(t)re^{at}}}{m!\Gamma(m+q+1)}\cdot\frac{\Gamma(m-q)}{\big(\frac{4a}{\delta^2}+c(t)\big)^{m-q}}\;.
\end{align*}
Due to $\lim\limits_{t\to 0} c(t)^{q+1+2m}e^{-c(t)re^{at}}=0$, the above power series is integrable over $(0,s)$ for every $s\in \R_+$. 
\end{proof}
Throughout this paper we will use the following notation:
for a fixed $r^*\in\R_+$, we define
\begin{align*}
&\tau:=\inf\{t \ge 0: r_t=r^*,\; r_0=r\le r^*\}
\\&\rho:=\inf\{t \ge 0: r_t=r^*,\; r_0=r\ge r^*\}\;.
\end{align*}
Further, we let
\begin{align}
&\psi_1(r):=\mE_r\Big[\int_0^{\tau}e^{-r_s}\md s\Big], \mbox{for $r\le r^*$},\label{intro:def1}
\\&\phi_1(r):=\mE_r\Big[\one_{[\rho<\infty]}\Big], \mbox{for $r\ge r^*$},\label{intro:def2}
\\&\phi_2(r):=\mE_r\Big[\one_{[\rho<\infty]}\rho\Big] \mbox{for $r\ge r^*$}.\label{intro:def3}
\end{align}
Since $\lim\limits_{t\to\infty} r_t=\infty$ a.s., we know $\tau<\infty$ a.s.
\begin{Lem}\label{intro:lem1}
The functions $\psi_1(r)$, $\phi_1(r)$ and $\phi_2(r)$ solve the differential equations
\begin{align}
&e^{-r}+(ar+b)g'(r)+\frac{\delta^2r}2 g''(r)=0, \label{eq3}
\\&(ar+b)g'(r)+\frac{\delta^2r}2g''(r)=0,\label{eq1}
\\&(ar+b)g'(r)+\frac{\delta^2r}2g''(r)+\phi_1(r)=0,\label{eq2}
\end{align}
correspondingly with boundary conditions
\begin{align*}
&\psi_1(r^*)=0 \mbox{ and } \psi_1'(0)=-\frac 1b,
\\&\phi_1(r^*)=1\mbox{ and } \phi_1(\infty)=0,
\\&\phi_2(r^*)=0\mbox{ and } \phi_2(\infty)=0. 
\end{align*}
\end{Lem}
For the proof confer the appendix, Section \ref{app}.
\section{Main Results}
Before considering an insurance company with surplus process following a Brownian motion, we look at the problem of consumption maximisation for an individual with a deterministic income. The discounting factor is assumed to be given by an exponential CIR process, $\{e^{-r_t}\}$. The filtration $\{\mF_t\}$ is generated by $\{r_t\}$
So, let the income process of the considered individual or household be given by
\[
X_t=x+\mu t\;,\quad \mbox{$\mu>0$.}
\]
Let $C$ denote the accumulated consumption process up to time $t$ and the ex-consumption income be given by
\[
X^C_t=x+\mu t-C_t\;.
\]
We call a strategy $C$ admissible if it is adapted to the filtration $\{\mF_t\}$, is non-decreasing and fulfils $C_0\ge 0$, $X_t^C\ge 0$ for all $t\in\R_+$, meaning in particular that $C_t\le x+\mu t$. In the following, we denote the set of all admissible strategies by $\mathfrak A$. Our target is to find an optimal consumption strategy, such that 
\[
\mE_{(r,x)}\Big[\int_0^\infty e^{-r_s}\md C_s\Big]\to\max!
\]
i.e. the expected discounted consumption is maximised. The following notation will be used throughout this section:
\begin{align*}
&V^C(r,x):=\mE_{(r,x)}\Big[\int_0^\infty e^{-r_s}\md C_s\Big],
\\&V(r,x):=\sup\limits_{C\in\mathfrak A} V^C(r,x)\;.
\end{align*}
The Hamilton--Jacobi--Bellman (HJB) equation can be motivated using the standard methods from stochastic control theory, so that we omit the detailed derivation and just refer to \cite[pp.\ 98,103]{hs} and references therein. The HJB turns out to consist of two partial differential equations with linear coefficients:
\begin{align}
\max\big\{\mu V_x+(ar+b) V_r+\frac{\delta^2r}{2}V_{rr},e^{-r}-V_x\big\}=0\;.\label{hjb1}
\end{align}
In particular, Proposition \ref{sec0:ver} below will illustrate that the HJB equation corresponds to the considered problem. 
\medskip
\\To simplify our considerations we introduce the following notation
\[
\mathcal L(f)(r,x)=\mu f_x(r,x)+(ar+b) f_r(r,x)+\frac{\delta^2r}{2}f_{rr}(r,x)
\]
for any appropriate function $f:=\R_+^2\to \R$.

In order to get an idea how the value function and the optimal strategy look like, we consider first the performance function corresponding to the strategy ``maximal spending''. This function is given by
\begin{align}
H(r,x):=x e^{-r}+\mu\mE_r\Big[\int_0^\infty e^{-r_s}\md s\Big]= x e^{-r}+\mu\int_0^\infty M(r,s)\md s\;,\label{maxpayout}
\end{align}
with $M(r,t)=\mE_r[e^{-r_t}]$ and using Fubini's theorem. The function $M$ fulfils $M\in \mathcal C^{2,1}(\R_+^2)$ and solves the differential equation, confer the appendix, Section \ref{app}:
\[
(ar+b)M_r(r,t)+\frac{\delta^2 r}{2}M_{rr}(r,t)-M_t(r,t)=0\;.
\]
This in particular means that, using the Leibniz integral rule, we get
\[
H_r(r,x)=- x e^{-r}+\mu\int_0^\infty M_r(r,s)\md s\quad\mbox{and} \quad H_{rr}(r,x)=x e^{-r}+\mu\int_0^\infty M_{rr}(r,s)\md s\;.
\]
Inserting $H(r,x)$ into the HJB equation yields on the one hand $e^{-r}-H_x(r,x)=0$ and on the other hand using $\int_0^\infty M_s(r,s)\md s=M(r,s)\Big|_0^\infty=-e^{-r}$ and the differential equation for $M$ we obtain that: 
\begin{align*}
\mu H_x+(ar+b) H_r+\frac{\delta^2r}{2}H_{rr}&=\mu e^{-r}+xe^{-r}\big(-ar-b+\frac{\delta^2r}2\big)+\mu\\&{}\int_0^\infty (ar+b)M_r(r,s)+\frac{\delta^2 r}{2}M_{rr}(r,s)\md s
\\&=\mu e^{-r}+xe^{-r}\big(-ar-b+\frac{\delta^2r}2\big)+\mu\int_0^\infty M_s(r,s)\md s
\\&=\mu e^{-r}+xe^{-r}\big(-ar-b+\frac{\delta^2r}2\big)-\mu e^{-r}
\\&= xe^{-r}\big(-ar-b+\frac{\delta^2r}2\big)\;.
\end{align*}
Note that the sign of the above expression does not depend on $x$ and define for $\frac{\delta^2}2>a$
\begin{equation}
R:=\frac b{\frac{\delta^2}2-a}\;.\label{R}
\end{equation}
\\In the following we will consider different combinations of the parameters $a$, $b$ and $\delta$, influencing the solution to the HJB equation \eqref{hjb1}.
%%%%%%%%%%%%%%%%%%%%%%%%%%%%%%%%%%%%%%%%%%%%%%%%%%%%%%%%%%%%%%%%%%%%%%%%%%%%%%%%%%%%%%%%%
\subsection{The case $\frac{\delta^2}2\le a$\label{sub:1}}
In this case it obviously holds 
\[
-ar-b+\frac{\delta^2r}2<0
\]
independently of $b$ and $r$, which in turn means that $H(r,x)$ solves the HJB equation \eqref{hjb1}. 
%If $H(r,x)$ is the value function, the optimal strategy would be to always spend the maximal possible amount independent of $r$ and $x$.  
\\Now, we can formulate the following verification theorem:
\begin{Thm}\label{thm1}
The function $H(r,x)$ is the value function and the strategy $C^{\max}_t:=x+\mu t$ ``$\rightsquigarrow$ to always spend the maximal possible amount independent of $r$ and $x$'' is the optimal strategy.
\end{Thm}
We skip the proof, as it goes similar to the proof of the verification theorem in the next subsection. 
%%%%%%%%%%%%%%%%%%%%%%%%%%%%%%%%%%%%%%%%%%%%%%%%%%%%%%%%%%%%%%%%%%%%%%%%%%%%%%%%%%%%%%%%%
\subsection{The case $a<\frac{\delta^2}2$}
In this case, $H(r,x)$ defined in \eqref{maxpayout} does not solve the HJB equation \eqref{hjb1} for $r>R$. 
\\For instance in \cite[p.\ 27]{hs} one finds that in order to solve an optimisation problem there are two ways: to show directly that the value function solves the HJB equation or to guess the optimal strategy and to prove that the corresponding return function solves the HJB equation. Here, we will follow the second method.

We conjecture that the optimal strategy is of a barrier type, i.e. there is a positive constant $\bar r\in\R_+$ such that it is optimal to wait if $r>\bar r$ and to immediately spend everything if $r\le \bar r$. Since we do not know how the optimal barrier should look like, we let $\bar r\in\R_+$ be arbitrary but fixed. The corresponding return function consists of two parts:
%In order to prove this claim, we define
\begin{align*}
&F(r,x):=xe^{-r}+\mu\mE_r\Big[\int_0^\tau e^{-r_s}\md s\Big]+\tilde F,\quad r\le \bar r
\\&G(r,x):=\mE_r\Big[\big(x+\mu \rho\big)\one_{[\rho<\infty]}\Big]e^{-\bar r}+\tilde F\mE_r\Big[\one_{[\rho<\infty]}\Big],\quad r> \bar r\;.
\end{align*}
where $\tilde F$ is some positive constant whose value should be determined later. 
It means that $F$ describes the spendings if the initial value $r_0=r\le \bar r$, and $G$ describes the waiting until $r_t$ approaches $\bar r$ or $\infty$.
\\By construction, $G(\bar r,x)=F(\bar r,x)$ and $G_x(\bar r,x)=F_x(\bar r,x)$ for all $x\in\R_+$.

The question is whether $G$ and $F$ given above solve the HJB equation \eqref{hjb1} on $[\bar r,\infty)$ and on $[0,\bar r]$ respectively and fulfil $G_r(\bar r, x) = F_r(\bar r, x)$ and $G_{rr}(\bar r, x) = F_{rr}(\bar r, x)$ for all $x\in\R_+$ with bounded derivatives $F_r$ and $G_r$. 
\subsubsection{Properties of $F$ and $G$}
In this subsection, we investigate the properties of functions $F$ and $G$.
Using notation \eqref{intro:def1}, we can rewrite $F$ as follows
\[
F(r,x)=xe^{-r}+\mu\psi_1(r)+\tilde F\;.
\] 
That is, inserting the function $F$ into the HJB equation \eqref{hjb1} we obtain on the one hand $e^{-r}-F_x=0$ and on the other hand, using Lemma \ref{intro:lem1}:
\begin{align*}
\mathcal{L}(F)(r,x)&=\mu e^{-r}+xe^{-r}\big(-ar-b+\frac{\delta^2r}2\big)+\mu(ar+b)\psi_1'+\mu\frac{\delta^2 r}2\psi_1''(r)
\\&=xe^{-r}\big(-ar-b+\frac{\delta^2r}2\big)\;.
\end{align*}
Thus, $F$ solves the HJB equation \eqref{hjb1} on the set $[0,\bar r]\times\R_+$, if $\bar r\le R$, defined in \eqref{R}.
\\Consider now the function $G$. Using Definitions \eqref{intro:def2} and \eqref{intro:def3}, $G$ can be rewritten as follows
\[
G(r,x)=x\phi_1(r)e^{-\bar r}+\mu \phi_2(r)e^{-\bar r}+\tilde F\phi_1(r)\;.
\]
We are going to find out under which conditions $G$ solves the HJB equation \eqref{hjb1} on the interval $[\bar r,\infty)$. Lemma \ref{intro:lem1} yields 
\begin{align*}
\mathcal L(G)(r,x)&=\mu\phi_1(r)e^{-\bar r}+\big(xe^{-\bar r}+\tilde F\big)\Big\{(ar+b)\phi_1'(r)+\frac{\delta^2r}2\phi_1''(r)\Big\}
\\&\quad {}+ \mu e^{-\bar r}\Big\{(ar+b)\phi_2'(r)+\frac{\delta^2r}2\phi_2''(r)\Big\}
\\&= \mu\phi_1(r)e^{-\bar r}-\mu \phi_1(r)e^{-\bar r}=0\;.
\end{align*}
Therefore, we have to consider $e^{-r}-G_x(r,x)$ and search for conditions supplying the relation $e^{-r}-G_x(r,x)\le 0$ on $[\bar r,\infty)\times \R_+$ in order for $G$ to solve the HJB equation \eqref{hjb1}. 
First, we prove the following auxiliary result:
\begin{Lem}\label{lem:phi1}
The function $\phi_1(r)$ is decreasing, and there is a unique $r^*\in[0,R]$ such that
\[
\phi_1(r^*)=-\phi_1'( r^*)=1\quad\mbox{and}\quad \phi_1(r)>-\phi_1'(r)\quad \mbox{for $r>r^*$}\;.
\]
\end{Lem}
\begin{proof}
For the proof confer the appendix, Section \ref{app}.
\end{proof}
The following Lemma considers the expression $e^{-r}-G_x(r,x)$ if the barrier is given by $r^*$ defined above.
\begin{Lem}
Let $\bar r = r^*$, defined in Lemma \ref{lem:phi1}. Then, for all $r> r^*$ the following inequality holds true:
\[
G_x(r,x)=\phi_1(r)e^{-r^*}> e^{-r} \;.
\]
\end{Lem}
\begin{proof}
Deriving $e^r\phi_1(r)$ yields, using Lemma \ref{lem:phi1}:
$\big(e^r\phi_1(r)\big)'=e^r\phi_1(r)\Big(1+\frac{\phi_1'(r)}{\phi_1(r)}\Big)> 0$.
Then,
\[
e^{-r}-G_x(r,x)=e^{-r}e^{-r^*}\Big(e^{r^*}-e^r\phi_1(r)\Big)<e^{-r}e^{-r^*}\Big(e^{r^*}-e^{r^*}\phi_1(r^*)\Big)=0
\]
for $r>r^*$.
\end{proof}
We can conclude that $G$ solves the HJB \eqref{hjb1} on $[r^*,\infty)$ if $\bar r=r^*$.
%%%%%%%%%%%%%%%%%%%%%%%%%%%%%%%%%%%%%%%%%%%%%%%%%%%%%%%%%%%%%%%%%%%%%%%%%%%%%%%%%%%%%%%%%
\subsubsection{The optimal strategy and verification theorem}
From now on we assume $\bar r=r^*$, i.e. $\frac{\phi_1'(r^*)}{\phi_1(r^*)}=-1$. 
\\
Due to Lemma \ref{lem:phi1}, both functions, $F$ and $G$, solve the HJB equation \eqref{hjb1} on $[0,r^*]$ and on $[r^*,\infty)$ correspondingly.
Next, we have to look at the derivatives $G_r(r^*,x)$, $F_r(r^*,x)$ and $G_{rr}(r^*,x)$, $F_{rr}(r^*,x)$ in order to guarantee a smooth value function. It holds
\begin{align*}
G_r(r^*,x)&=xe^{-r^*}\phi_1'(r^*)+\mu \phi_2'(r^*)e^{-r^*}+\tilde F\phi_1'(r^*)
\\&=-xe^{-r^*}+\mu \phi_2'(r^*)e^{-r^*}-\tilde F,
\\F_r(r^*,x)&=-xe^{-r^*}+\mu\psi_1'(r^*)
\end{align*}
and
\begin{align*}
G_{rr}(r^*,x)&= xe^{-r^*}\phi_1''(r^*)+\mu \phi_2''(r^*)e^{-r^*}+\tilde F\phi_1''(r^*),
%\\&= xe^{-r^*}+\mu \phi_2''(r^*)e^{-r^*}+\tilde F
\\F_{rr}(r^*,x)&= xe^{-r^*}+\mu\psi_1''(r^*)\;.
\end{align*}
\begin{Bem}
Choosing 
\begin{equation}
F(r^*,0)=\tilde F=\mu \phi_2'(r^*)e^{-r^*}-\mu\psi_1'(r^*)\label{smoothfit}
\end{equation}
yields $G_r(r^*,x)=F_r(r^*,x)$ for all $x\in\R_+$. Note that it holds $\tilde F\ge 0$ due to the proof of Lemma \ref{intro:lem1}, confer appendix, Section \ref{app}.
\\Consider the differential equations \eqref{eq3} multiplied by $(-\mu)$; \eqref{eq1} multiplied by $\tilde F$ and \eqref{eq2} multiplied by $\mu e^{-r^*}$ at $r^*$:
\begin{align*}
&-\mu e^{-r^*}-(ar^*+b)\mu\psi_1'(r^*)-\frac{\delta^2r^*}2 \mu\psi_1''(r^*)=0,
\\&(ar^*+b)\tilde F\phi_1'(r^*)+\frac{\delta^2r^*}2 \tilde F\phi_1''(r^*)=0,
\\&e^{-r^*}(ar^*+b)\mu\phi_2'(r^*)+\mu e^{-r^*}\frac{\delta^2r^*}2 \phi_2''(r^*)+\mu e^{-r^*}\phi_1(r^*)=0.
\end{align*} 
Using $\phi_1(r^*)=-\phi_1'(r^*)=1$ and adding the above equations yields
\[
\begin{split}
(ar^*+b)\big\{\mu\phi_2'(r^*)e^{-r^*}-\mu\psi_1'(r^*)&-\tilde F\big\}\\&=-\frac{\delta^2r^*}2\big\{\mu\phi_2''(r^*)e^{-r^*}+\tilde F\phi_1''(r^*)-\mu\psi_1''(r^*)\big\}\;.
\end{split}
\]
Note that by definition of $\tilde F$, the lhs of the above equation equals zero, meaning \\$G_{rr}(r^*,0)=F_{rr}(r^*,0)$. However, in general it does not hold $\phi_1''(r^*)=1$. For that reason $G_{rr}(r^*,x)\neq F_{rr}(r^*,x)$ if $x\neq 0$.
\end{Bem}
We formulate the following verification theorem.
\begin{Thm}\label{sec0:ver}
The optimal strategy $C^*$ is to immediately spend any available amount bigger than zero if $r\le r^*$, i.e. $C^*_t= \big(x+\mu \lambda_{r^*}^t\big)\one_{[\lambda_{r^*}^t> 0]}$, where $\lambda_{r^*}^t:=\sup\{s\in[0,t):\; r_s\le r^*\}$ with $\sup\{\emptyset\}=0$. The value function $V(r,x)$ solves the HJB equation \eqref{hjb1} and fulfils $V(r,x) =v(r,x)$ with
\[
v(r,x)=\begin{cases}
F(r,x) & \mbox{if $(r,x)\in[0,r^*]\times \R_+$}
\\G(r,x)& \mbox{ if $(r,x)\in[r^*,\infty)\times \R_+$}
\end{cases}
\]
with $F(r^*,0)=\tilde F$ given in \eqref{smoothfit}.
\end{Thm}
\begin{proof}
Note that it holds
\[
\{\lambda_{r^*}^t\le u\}=\big\{\sup\{s\in[0,t):\; r_s\le r^*\}\le u\big\}=\big \{\inf\limits_{u<  s\le t}r_s> r^*\big\}
\]
Because the running infimum above is $\mF_t$-measurable, we can conclude that the strategy $C^*$ defined above is an admissible strategy. 
\\Since $F\in \mathcal C^{1,2}((0,r^*)\times\R_+)$, $G\in \mathcal C^{1,2}((r^*,\infty)\times\R_+)$, $F(r^*,x)=G(r^*,x)$, $F_r(r^*,x)=G_r(r^*,x)$ and $\one_{[r_t=r*]}=1$ a.s., we can apply the change-of-variable formula due to \cite{peskir}. Let $C$ be an arbitrary admissible strategy and $\hat X$ the ex-consumption process under $C$. Then
\begin{align*}
v(r_t,\hat X_t)&=v(r,x)+\int_0^t \mathcal L(v)(r_s,\hat X_s)\md s+\int_0^t \delta\sqrt{r_s}\;v_r(r_s,\hat X_s)\md W_s
-\int_0^t v_x(r_s,\hat X_s) \md C_s.
\end{align*}
Further, we know that $v$ solves the HJB equation \eqref{hjb1}, meaning $\mathcal L(v)(r,x)\le 0$ and $e^{-r}-v_x(r,x)\le 0$ for all $(r,x)\in \R_+^2$. Therefore,
\begin{align*}
v(r_t,\hat X_t)\le v(r,x)+ \int_0^t \delta\sqrt{r_s}\;v_r(r_s,\hat X_s)\md W_s-\int_0^t e^{-r_s} \md C_s.
\end{align*}
The stochastic integral above is a martingale with expectation zero, confer the proof of Lemma \ref{intro:lem1} in Section \ref{app}. Taking the expectations on the both sides of the above equality, one gets
\[
\mE_{(r,x)}\big[v(r_t,\hat X_t)\big]\le v(r,x)-\mE_{(r,x)}\Big[\int_0^t e^{-r_s} \md C_s\Big].
\]
Because $\lim\limits_{r\to\infty}v(r,x)=0$ and $v(r,x)$ is bounded, by dominated convergence we can interchange limit and integration and obtain
\[
v(r,x)\ge \mE_{(r,x)}\Big[\int_0^\infty e^{-r_s} \md C_s\Big]
\]
Taking the strategy $C^*$ yields equality.
\end{proof}
Thus, the barrier strategy with barrier given by $r^*$, defined in Lemma \ref{lem:phi1}, is optimal. The corresponding return function is the value function and solves the HJB equation \eqref{hjb1}.
\begin{figure}[t]
\includegraphics[scale=0.4, bb = 0 0 200 400]{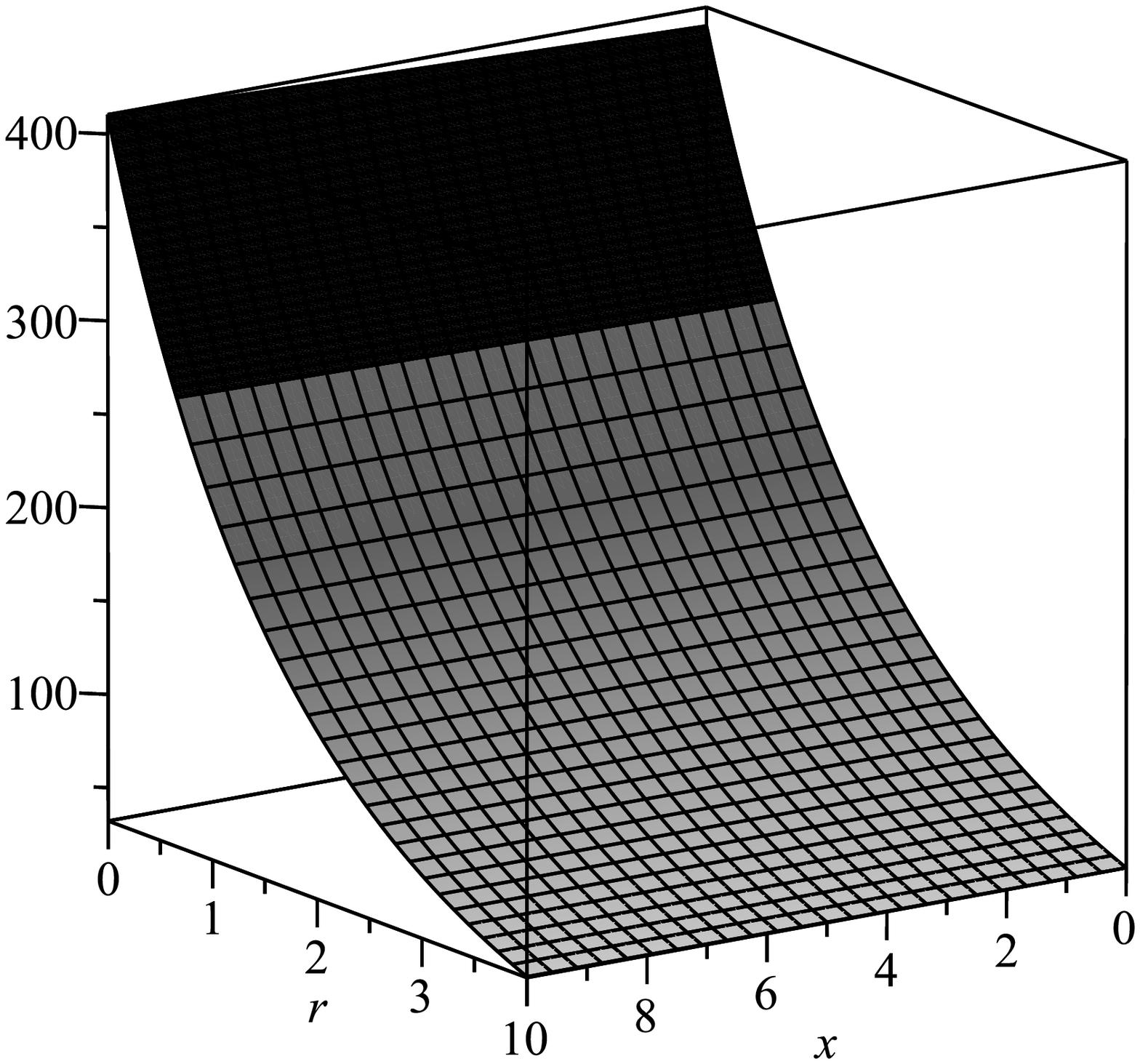}
\includegraphics[scale=0.33, bb = -450 -50 200 400]{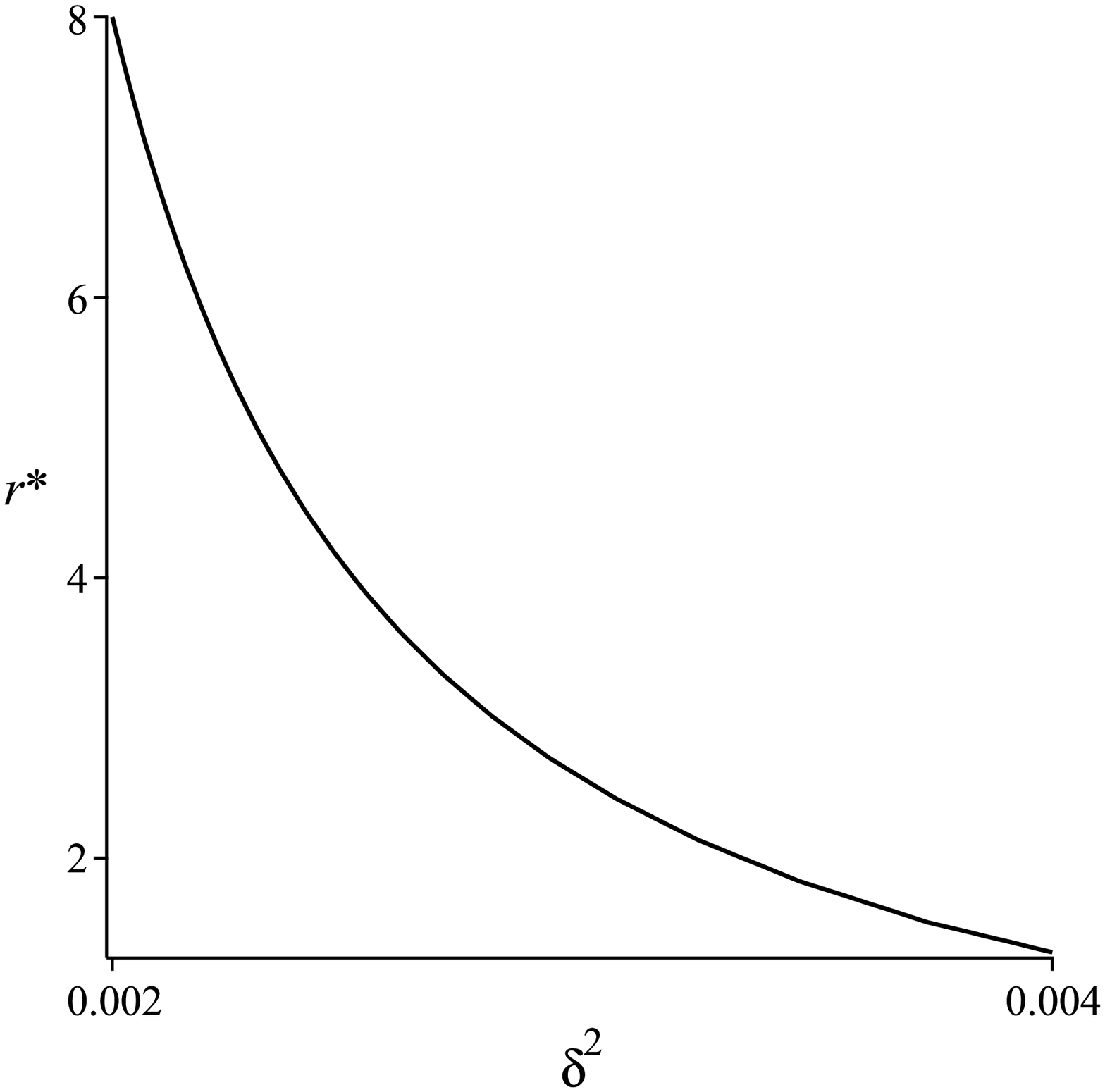}
\caption{Lhs: The value function $V(r,x)$, consisting of $F(r,x)$ (black) and $G(r,x)$ (grey). Rhs: Dependence of the barrier $r^*$ on $\delta^2$.\label{fig2}}
\end{figure}
\begin{Bsp}
We can calculate the value function explicitly:
\begin{align*}
&\psi_1(r)=\frac 2{\delta^2}\int_x^{r^*}y^{-\frac{2b}{\delta^2}}e^{-\frac{2a}{\delta^2}y}\int_0^ye^{-z}z^{\frac{2b}{\delta^2}-1}e^{\frac{2a}{\delta^2}z}\md z\md y,
\\&\phi_1(r)=\frac1{\int_{r^*}^\infty y^{-\frac{2b}{\delta^2}}e^{-\frac{2a}{\delta^2}y}\md y}\int_{r}^\infty y^{-\frac{2b}{\delta^2}}e^{-\frac{2a}{\delta^2}y}\md y,
\\&\phi_2(r)=\frac 2{\delta^2}\frac{\int_{r^*}^\infty y^{-\frac{2b}{\delta^2}}e^{-\frac{2a}{\delta^2}y}\int_{r^*}^y\phi_1(z)z^{\frac{2b}{\delta^2}-1}e^{\frac{2a}{\delta^2}z}\md z\md y}{\int_{r^*}^\infty y^{-\frac{2b}{\delta^2}}e^{-\frac{2a}{\delta^2}y}\md y}\int_{r^*}^r y^{-\frac{2b}{\delta^2}}e^{-\frac{2a}{\delta^2}y}\md y
\\&{}-\frac 2{\delta^2}\int_{r^*}^r y^{-\frac{2b}{\delta^2}}e^{-\frac{2a}{\delta^2}y}\int_{r^*}^y\phi_1(z)z^{\frac{2b}{\delta^2}-1}e^{\frac{2a}{\delta^2}z}\md z\md y.
\end{align*}
For well-definiteness of $\phi_1$ and $\phi_2$ confer the proof of Lemma 1.7 in the appendix, Section \ref{app}.
Let $a=0.001$, $b=0.002$, $\delta=0.07$ and $\mu=0.5$. 
The both functions, $F$ and $G$ are illustrated in Figure \ref{fig2}. 
\end{Bsp}
\subsection{The $0$-barrier strategy}
In the following, we are going to discuss a very special strategy for the case $2b<\delta^2$: ``spending just if the underlying CIR hits zero''. In fact, we know from \cite{cir} that if $2b<\delta^2$, a CIR process can attain zero with a positive probability. With growing volatility, the probability to ``dive'' and touch zero increases.

In the previous subsection, we have shown that the value function solves the HJB equation \eqref{hjb1} and that the optimal strategy is a barrier strategy with a constant barrier $r^*$ fulfilling $r^*\le R$, with $R$ given in \eqref{R}. By the structure of $R$, it holds $R\to0$ as $\delta\to\infty$. 
This brings up the question whether the optimal barrier could be equal to zero for some big values of $\delta$. 
\\Let $\delta<\infty$ be fixed and fulfil $\delta^2>2\max\{a,b\}$, meaning that the probability for $\{r_t\}$ to hit zero is positive. Due to subsection \ref{sub:1} the function $H$ defined in \eqref{maxpayout} is not the value function.
\\Consider the strategy $C^0$ with $C^0_t:=\big(x+\mu \lambda_0^t\big)\one_{[\lambda_0^t>0]}$, $\lambda_0^t=\sup\{s\in[0,t): r_s=0\}$ with $\sup\{\emptyset\}=0$ ``$\rightsquigarrow$ to spend the maximal possible amount only if $r_t=0$''.
\\
Letting again
\[
\rho_0:=\inf\{t \ge 0:\; r_t=0,\; r_0>0\},
\]
we define $\hat\phi_1(r):=\mE_r[\rho_0<\infty]$ and $\hat\phi_2(r):=\mE_r[\rho_0\one_{[\rho_0<\infty]}]$.
\begin{Lem}
The functions $\hat\phi_1(r)$ and $\hat\phi_2(r)$ solve the differential equations \eqref{eq1} and \eqref{eq2} on $(0,\infty)$ correspondingly. The function $\phi_2(r)$ is finite on $(0,\infty)$ and it holds
\begin{align*}
\hat\phi_1(r)&=\frac{1}{\int_{0}^\infty y^{-\frac{2b}{\delta^2}} e^{-\frac{2a}{\delta^2}y}\md y}\int_x^\infty y^{-\frac{2b}{\delta^2}} e^{-\frac{2a}{\delta^2}y}\md y,
\\\hat\phi_2(r)&=\frac 2{\delta^2}\frac{\int_{0}^\infty y^{-\frac{2b}{\delta^2}}e^{-\frac{2a}{\delta^2}y}\int_{0}^y\phi_1(z)z^{\frac{2b}{\delta^2}-1}e^{\frac{2a}{\delta^2}z}\md z\md y}{\int_{0}^\infty y^{-\frac{2b}{\delta^2}}e^{-\frac{2a}{\delta^2}y}\md y}\int_{0}^r y^{-\frac{2b}{\delta^2}}e^{-\frac{2a}{\delta^2}y}\md y
\\&{}-\frac 2{\delta^2}\int_{0}^r y^{-\frac{2b}{\delta^2}}e^{-\frac{2a}{\delta^2}y}\int_{0}^y\phi_1(z)z^{\frac{2b}{\delta^2}-1}e^{\frac{2a}{\delta^2}z}\md z\md y
\end{align*}
\end{Lem}
\begin{proof}
Due to Lemma \ref{lem:endl} it holds $\int_0^\infty y^{-\frac{2b}{\delta^2}} e^{-\frac{2a}{\delta^2}y}\md y<\infty$.
Due to Lemma \ref{intro:lem1}, the functions $\hat\phi_1(r)=\mE_r[\rho<\infty]$ and $\hat\phi_2(r)=\mE_r[\rho\one_{[\rho<\infty]}]$ solve the differential equations \eqref{eq1} and \eqref{eq2} on $(0,\infty)$ correspondingly and the function $\phi_2$ is finite.
\end{proof} 
Define further
\[
\lambda_0 := \sup\{t \ge 0:\; r_t=0\}\;,
\]
i.e. $\lambda_0$ is the last exit time from zero before the $\{r_t\}$ approaches $\infty$. It is clear that using $C^0$, in case $r_0=0$ one spends everything immediately, saves money until $\{r_t\}$ approaches zero for the next time and spends everything there. The game ends at time $\lambda_0$ defined above. 
\begin{Lem}\label{sec0:lastexit}
Let $\frac{\delta^2}2>b$, then
\[
\lambda_0<\infty \quad\mbox{a.s.} \quad\mbox{and}\quad \mE_0[\lambda_0]=\int_0^\infty t \frac{\big(e^{at}-1\big)^{-\frac{2b}{\delta^2}}}{\int_0^\infty \big(e^{az}-1\big)^{-\frac{2b}{\delta^2}}\md z}\md t <\infty\;.
\]
\end{Lem} 
\begin{proof}
For the proof confer the appendix, Section \ref{app}.
\end{proof}
Now, we can write down the return function corresponding to the strategy $C^0$:
\[
V^0(r,x):=\mE_r\big[\big(x+\mu \rho_0+\tilde V^0\big)\one_{[\rho_0<\infty]}\big]=(x+\tilde V^0)\hat\phi_1(r)+\mu\hat\phi_2(r)\;,
\]
where $\tilde V^0=V^0(0,0)=\mu\mE_0[\lambda_0]$.
\begin{Satz}
The strategy $C^0$ cannot be optimal for any $\delta<\infty$. 
\end{Satz}
\begin{proof}
Recall that the function $\hat\phi_1(r)$ is given by
\[
\hat\phi_1(r)=\frac{1}{\int_{0}^\infty y^{-\frac{2b}{\delta^2}} e^{-\frac{2a}{\delta^2}y}\md y}\int_x^\infty y^{-\frac{2b}{\delta^2}} e^{-\frac{2a}{\delta^2}y}\md y\;.
\]
This means in particular $\hat\phi_1(0)=1$ and $\lim\limits_{r\to0}\hat\phi_1'(r)=-\infty$ giving 
\[
e^{-r}-V^0_x(r,x)=e^{-r}-\phi_1(r)>0
\] 
for $r\in (0,\varepsilon)$ and some $\varepsilon>0$. This means that $V^0$ does not solve the HJB \eqref{hjb1} on $(0,\varepsilon)\times\R_+$. Since in the previous subsection we have shown that the value function solves the HJB, we can conclude that $C^0$ will never be optimal.
\end{proof}
%%%%%%%%%%%%%%%%%%%%%%%%%%%%%%%%%%%%%%%%%%%%%%%%%%%%%%%%%%%%%%%%%%%%%%%%%%%%%%%%%%%%%%

\subsection{The Brownian risk model}
\noindent
In this subsection, we add complexity to our model by assuming that the underlying surplus (previously called income) process is given by a Brownian motion with drift. Since it is unrealistic to assume strong random fluctuations in the income of an individual or household, we change the economic interpretation from maximising the consumption of an individual to the maximising of dividends of an insurance company.
The difference to the previous case comes up also in the fact that we stop our considerations when the surplus becomes negative (ruins). Taking into consideration the ruin time destroys the linear dependence of the value function on the surplus. In general, the return function corresponding to a constant barrier strategy will have a representation as a power series with non-linear functions as summands.
Therefore - using the chess terminology - in order to keep the problem in check, we assume for this subsection $\delta^2= 2a$.

Technically, we consider an insurance company whose surplus is given by a Brownian motion with drift $X_t=x+\mu t+\sigma B_t$, where $\{B_t\}$ is a standard Brownian motion and $\mu,\sigma>0$ are positive constants. The considered insurance company is allowed to pay out dividends, where the accumulated dividends until $t$ are given by $C_t$, yielding for the ex-dividend surplus $X^C$:
\[
X_t^C=x+\mu t+\sigma B_t-C_t\;.
\]
The consideration will be stopped at the ruin time $\tau^C$ of $X^C$. Let further $\{W_t\}$, the Brownian motion driving the discounting CIR process \eqref{intro:cir}, be independent of $\{B_t\}$, and the underlying filtration $\{\mF_t\}$ be the filtration generated by the pair $\{W_t,B_t\}$.  
We call a strategy $C$ admissible if $C_t$ is adapted to $\{\mF_t\}$, $C_0\ge 0$ and $X_t^C\ge 0$ for all $t\ge 0$, the set of admissible strategies will be denoted by $\mathfrak B$.
\\As a risk measure we consider the value of expected discounted dividends, where
the dividends are discounted by an CIR process \eqref{intro:cir}.

We define the return function corresponding to some admissible strategy $C$ to be $V^C(r,x)=\mE_{(r,x)}\Big[\int_0^{\tau^C} e^{-r_s}\md C_s\Big]$ and let
\begin{align*}
V(r,x)=\sup\limits_{C\in\mathfrak B}V^C(r,x)\;.
%\\&\lambda:=\frac{2a}{\delta^2}\;.
\end{align*}
The HJB equation corresponding to the problem can be derived similarly to \cite[pp.\ 98,103]{hs}:
\begin{align}
\max\big\{\mu V_x+\frac{\sigma^2}2V_{xx}+(ar+b) V_r+ ar V_{rr},e^{-r}-V_x\big\}=0\;.\label{hjb2}
\end{align} 
In this setup, we conjecture that the optimal strategy will be of a barrier type with a constant barrier for the surplus process. It means, we pay any capital bigger than the barrier, independent of $\{r_t\}$.
Define now the following auxiliary quantities:
\begin{equation*}
\begin{split}
&\theta:=\frac{-\mu+\sqrt{\mu^2 +2\sigma^2 b}}{\sigma^2},\quad \zeta:=\frac{-\mu-\sqrt{\mu^2 +2\sigma^2 b}}{\sigma^2},\quad \varrho:=\frac{\ln\big(b-\mu\zeta\big)-\ln\big(b-\mu\theta \big)}{\theta-\zeta}\;.
\end{split}
\end{equation*}
\begin{Lem}\label{lem:brow}
The return function $V^\varrho(r,x)$ corresponding to the constant barrier strategy $\varrho$ is given by 
\[
V^\varrho(r,x)=\begin{cases}F(r,x) &\mbox{: $x\ge \varrho$}\\
G(r,x) &\mbox{: $x\le \varrho$},
\end{cases}
\]
where
\begin{align*}
&F(r,x):=\big(x-\varrho+\frac\mu b\big)e^{-r},\quad\mbox{if $x\ge \varrho$}
\\&G(r,x):=e^{-r}\frac{e^{\theta x}-e^{\zeta x}}{\theta e^{\theta \varrho}-\zeta e^{\zeta \varrho}} ,\quad\mbox{if $x\le \varrho$}.
\end{align*}
The functions $F$ and $G$ fulfil:
\begin{itemize}
\item $F(r,\varrho)=G(r,\varrho)$, $F_r(r,\varrho)=G_r(r,\varrho)$, $F_{rr}(r,\varrho)=G_{rr}(r,\varrho)$;
\item $F_x(r,\varrho)=e^{-r}=G_x(r,\varrho)$ and $F_{xx}(r,\varrho)=0=G_{xx}(r,\varrho)$ for all $r\in\R_+$;
\item $G(r,x)$ solves the partial differential equation
\[
\mu f_x+\frac{\sigma^2}2f_{xx}+(ar+b) f_r+ ar f_{rr}=0
\]
and fulfils $G_x(r,x)\ge e^{-r}$ for all $r\in\R_+$ and $x\in[0,\varrho]$;
\end{itemize}
\end{Lem}
\begin{proof}
For the proof confer \cite{shreve}, Lemma 2.1 and Corollary 2.2.
\end{proof}
\begin{Satz}
The optimal dividend strategy $C^*$ is to pay any capital larger than $\varrho$, i.e. $$C_t^*=\max\Big\{\sup\limits_{0\le s\le \tau^*\w t}\Big(x+\mu s+\sigma W_s\Big)-\varrho;0\Big\},$$ where $\tau^*$ is the ruin time. The value function $V(r,x)$ is given by $F$ on $[\varrho,\infty)$, by $G$ on $[0,\varrho]$ and solves the HJB equation \eqref{hjb2}.
\end{Satz}
\begin{proof}
Using Lemma \ref{lem:brow}, the proof follows closely the proof in \cite{astak}, see also \cite[p.\ 104]{hs}.
\end{proof}
Thus, if $\delta^2=2a$, i.e. $e^{bt}e^{-r_t}$ is a martingale (confer Lemma \ref{intro:convex} and Definition \eqref{mef}), the optimisation problem can be reduced to the classical dividend optimisation problem with a constant discounting rate, described in \cite{astak}. 
\subsection{Conclusion}
For the deterministic income, we considered different cases dependent on the relation between the parameters $\delta^2$ and $a$. In both cases, the optimal strategy turns out to be of a barrier type, i.e. it is optimal to spend all available money only if the process $\{r_t\}$ is below a certain level, otherwise it is optimal to wait.
\\If the volatility coefficient $\delta$ is relatively small, i.e. $\delta^2\le 2a$, then the paths are going ``nearly deterministically'' to infinity, meaning that $e^{-r_t}$ is a supermartingale. Therefore, the optimal barrier is lying at infinity, and it is always optimal to spend the maximal possible amount.
\\If $2a<\delta^2$, the process $\{r_t\}$, moving $\delta^2$ from $2a$ upwards shifts the optimal barrier from $\infty$ to $0$, see Figure \ref{fig2}. It means, the higher the volatility the likely the process can hit a lower level. It makes sense to wait until the discounting process attains ``small'' values, and spend the saved amount there.
\\Finally, we showed that the strategy ``spendings only if $r_t=0$'' is never optimal, i.e. the optimal barrier is always greater than zero. 
\medskip
\\Note that the above results strongly differ from the case of integrated Ornstein-Uhlenbeck (OU) discounting. There, see \cite{eis1}, it was optimal to wait if the interest rate was below a certain level and to start consuming otherwise. The reason for swapping of the paying behaviour in the case of an CIR discounting roots in the fact that an OU process can attain negative values. For more details confer \cite{eis1}.

In the Brownian risk model, the case $2a\neq \delta^2$ has not been considered and is a subject to future research. We conjecture that the optimal strategy there will be of a barrier type with a non-constant barrier depending on the underlying CIR process.  
\section{Appendix\label{app}}
\small
\subsubsection*{Proof of Lemma \ref{app:lem1}} 
Assume there exists a set $A\in\mF$ with $\mP[A]>0$ and $\liminf\limits_{t\to\infty} r_t=B<\infty$ on $A$. Then, there is a sequence $t_n\to\infty $ as $n\to\infty$ such that $\lim\limits_{n\to\infty} r_{t_n}=B$ on $A$. 
By Lebesgue's dominated convergence theorem and using $\lim\limits_{t\to\infty}\mE[e^{-r_t}]=\lim\limits_{t\to\infty}M(r,t)=0$, confer \eqref{mef} for definition of $M$, we obtain
\[
0=\lim\limits_{n\to\infty}\mE_r[e^{-r_{t_n}}]\ge \lim\limits_{n\to\infty}\mE[e^{-r_{t_n}}\one_{A}]=e^{-B}\mP[A]>0.
\]
The last inequality is a contradiction proving our claim.
\hfill $\square$
\subsubsection*{Proof of Lemma \ref{intro:lem1}}
\noindent
\textbf{Part I:}
\\Due to \cite[p.\ 127]{walter}, the differential equation 
\[
e^{-r}+(ar+b)g'(r)+\frac{\delta^2r}2 g''(r)=0
\]
has twice continuously differentiable solutions on $[0,r^*]$. 
A general solution to the above differential equation is given by
\begin{align*}
g'(r)=\Big( \frac {-2}{{\delta}^{2}}\int \!{{y}^{{\frac {-{\delta}^{2}+
2\,b}{{\delta}^{2}}}}{ e^{ \left( {\frac {2a}{{\delta}^{2}}}-1
 \right) y}}}\,{\md}y+C \Big) { e^{-\frac {2a}{{
\delta}^{2}}r}}{r}^{-{\frac {2b}{{\delta}^{2}}}}\;.
\end{align*}
Therefore, in order to have $g'(0)>-\infty$ we must define
\[
g'(r)=\Big(-\frac2{\delta^2}\int_0^r y^{\big(\frac{2b}{\delta^2}-1\big)}e^{\big(\frac{2a}{\delta^2}-1\big)y}\md y\Big) r^{-\frac{2b}{\delta^2}}e^{-\frac{2a}{\delta^2}r}\;.
\]
Now, letting $r\to 0$ and using L'Hospital's rule:
\[
\lim\limits_{r\to 0} g'(r)=-\frac 1b\;.
\]
Let $\tilde\psi_1(r)$ denote the unique solution with boundary conditions $\tilde\psi_1(r^*)=0$ and $\tilde\psi_1'(0)=-\frac 1b$. In this case it holds
\[
\lim\limits_{r\to\infty}r\tilde\psi_1''(r)=0\;,
\]
which means $\tilde\psi_1''(r)\in o(\frac 1r)$ for $r\to 0$. Thus, we can apply Ito's formula on $\tilde\psi_1(r_{\tau\w t})$:
\begin{align*}
\tilde\psi_1(r_{\tau\w t})=\tilde\psi_1(r)+\int_0^{\tau\w t}(ar_s+b)\tilde\psi_1'(r_s)+\frac{\delta^2r_s}2\tilde\psi_1''(r_s)\md s +\int_0^{\tau\w t} \delta\sqrt{r_s}\tilde\psi_1'(r_s) \md W_s\;.
\end{align*}
Since $\tilde\psi_1'$ is bounded, 
%fulfils $\int_0^t \mE[\sqrt{r_s}\tilde\phi_1'(r_s)]\md s<\infty$, \cite[p.\ 130\ Corollary 1.25]{revuzyor}.
the stochastic integral is a martingale with expectation zero. Therefore, taking the expectations on the both sides and letting $t\to\infty$ (interchanging of expectations and limit is possible due to the bounded convergence theorem) we obtain
\begin{align*}
\tilde\psi_1(r)=\mE_r\Big[\int_0^{\tau}e^{-r_s}\md s\Big]=\psi_1(r)\;.
\end{align*}
\textbf{Part II:}
\\It is clear that if $2b\ge\delta^2$ and $r^*=0$, we have $\phi_1(r)=\one_{\{0\}}$. 
Therefore, we just need to consider the remaining cases. Differential equation \eqref{eq1} has a unique solution on $(r^*,\infty)$, say $\tilde\phi_1(r)$, with boundary conditions $\tilde\phi_1(r^*)=1$ and $\tilde\phi_1(\infty)=0$:
\[
\tilde \phi_1(r)=\frac1{\int_{r^*}^\infty y^{-\frac{2b}{\delta^2}}e^{-\frac{2a}{\delta^2}y}\md y}\int_{r}^\infty y^{-\frac{2b}{\delta^2}}e^{-\frac{2a}{\delta^2}y}\md y\;.
\] 
Applying Ito's formula on $\tilde\phi_1$ yields
\begin{align}
\tilde\phi_1(r_{\rho\w t})=\tilde\phi_1(r)+\int_0^{\rho\w t}(ar+b)\tilde\phi_1'(r_s)+\frac{\delta^2r}2\tilde\phi_1''(r_s)\md s+\int_0^{\rho\w t}\delta\sqrt{r_s}\tilde\phi_1'(r_s)\md W_s\;.\label{app:phi1}
\end{align} 
If $r^*>0$, then $\sqrt{r_{\rho\w s}}\tilde \phi'_1(r_{\rho\w s})$ is bounded and the stochastic integral is a martingale with expectation zero. 
\\If $r^*=0$ and $2b<\delta^2$ then:
\[
\Big(\sqrt{r_s}\tilde\phi_1'(r_s)\Big)^2=r_s^{1-\frac{4b}{\delta^2}}e^{-\frac{4a}{\delta^2}r_s}\frac1{\Big(\int_{0}^\infty y^{-\frac{2b}{\delta^2}}e^{-\frac{2a}{\delta^2}y}\md y\Big)^2}. 
\] 
Note that $\int_{0}^\infty y^{-\frac{2b}{\delta^2}}e^{-\frac{2a}{\delta^2}y}\md y<\infty$ and
\[
\int_0^t\mE\Big[\Big(\sqrt{r_s}\tilde\phi_1'(r_s)\Big)^2\Big]\md s<\infty
\]
for all $t\in\R_+$ due to Lemma \ref{lem:endl}. Then due to \cite[p.\ 130\ Corollary 1.25]{revuzyor}, the stochastic integral in \eqref{app:phi1} is a martingale with expectation zero.
Applying the expectations and letting $t$ go to infinity in \eqref{app:phi1}, one obtains 
\[
\tilde\phi_1(r)=\mE\Big[\one_{[\rho<\infty]}\Big]=\phi_1(r)\;.
\]
\textbf{Part III:}
If $r^*=0$ and $2b\ge \delta^2$ then obviously $\phi_2(r)\equiv 0$. Consider now the remaining cases.
Differential equation \eqref{eq2} has a unique solution
\begin{align*}
\tilde\phi_2(r)&=\frac 2{\delta^2}\frac{\int_{r^*}^\infty y^{-\frac{2b}{\delta^2}}e^{-\frac{2a}{\delta^2}y}\int_{r^*}^y\phi_1(z)z^{\frac{2b}{\delta^2}-1}e^{\frac{2a}{\delta^2}z}\md z\md y}{\int_{r^*}^\infty y^{-\frac{2b}{\delta^2}}e^{-\frac{2a}{\delta^2}y}\md y}\int_{r^*}^r y^{-\frac{2b}{\delta^2}}e^{-\frac{2a}{\delta^2}y}\md y
\\&\quad{}-\frac 2{\delta^2}\int_{r^*}^r y^{-\frac{2b}{\delta^2}}e^{-\frac{2a}{\delta^2}y}\int_{r^*}^y\phi_1(z)z^{\frac{2b}{\delta^2}-1}e^{\frac{2a}{\delta^2}z}\md z\md y
\end{align*}
with boundary conditions $\tilde\phi_2(r^*)=0=\tilde\phi_2(\infty)$. Note that for $r^*>0$ it holds due to the structure of $\phi_1$ given above:
\begin{align*}
\int_{r^*}^\infty y^{-\frac{2b}{\delta^2}}e^{-\frac{2a}{\delta^2}y}\int_{r^*}^y\phi_1(z)z^{\frac{2b}{\delta^2}-1}e^{\frac{2a}{\delta^2}z}\md z\md y\le \frac 1{r^*}\phi_1(r^*)<\infty\;.
\end{align*}
Let now $r^*=0$ and $2b<\delta^2$. Then, applying partial integration for the inner integral and using that the negative part is smaller than zero, we get
\begin{align*}
\int_{0}^\infty y^{-\frac{2b}{\delta^2}}e^{-\frac{2a}{\delta^2}y}\int_{0}^y\phi_1(z)z^{\frac{2b}{\delta^2}-1}e^{\frac{2a}{\delta^2}z}\md z\md y\le \frac{\delta^2}{2b}\int_0^\infty y^{-\frac{2b}{\delta^2}}e^{-\frac{2a}{\delta^2}}\Big[y^{\frac{2b}{\delta^2}}e^{\frac{2a}{\delta^2}}\phi_1(y)+y\Big]<\infty\;.
\end{align*}
The finiteness of the integral above follows from the properties of the Gamma distribution. Further, it is easy to see that $\tilde\phi_2'(r^*)>0$. Since $\tilde\phi_2$ solves the differential equation \eqref{eq2}, it follows immediately $\tilde\phi_2''(r)<0$ if $\tilde \phi_2'(r)=0$. This means in particular that after becoming negative, the derivative $\tilde\phi_2'(r)$ remains negative. Therefore $\tilde\phi_2(\infty)=0$ implies $\tilde\phi_2(r)\ge 0$.\smallskip
\\For the function $\tilde\phi_2(r)$ it holds
\begin{align*}
\tilde\phi_2(r_{\rho\w t})=\tilde\phi_2(r)+ \int_0^{\rho\w t}(ar_s+b)\tilde\phi_2'(r_s)+\frac{\delta^2r_s}2\tilde\phi_2''(r_s)\md s +\int_0^{\rho\w t}\delta\sqrt{r_s}\tilde\phi_2'(r_s)\md W_s\;.
\end{align*}
Similar to Part II, using Lemma \ref{lem:endl} and \cite[p.\ 130\ Corollary 1.25]{revuzyor} one can show that the stochastic integral above is a martingale with expectation zero.
Applying the expectations yields
\begin{align*}
\mE\Big[\tilde\phi_2(r_{\rho\w t})\Big]=\tilde\phi_2(r)-\mE\Big[\int_0^{\rho\w t}\phi_1(r_s)\md s\Big]\;.
\end{align*}
Note that applying Fubini's theorem on the expectation on the rhs, one obtains
\begin{align*}
\mE\Big[\int_0^{\rho\w t}\phi_1(r_s)\md s\Big]&=\int_0^\infty \mE\Big[\one_{[0\le s\le\rho\w t]}\phi_1(r_s)\Big]\md s
=\int_0^\infty \mE\Big[\one_{[0\le s\le\rho\w t]}\mP[\rho<\infty|r_s]\Big] \md s
\\&=\int_0^\infty \mE\Big[\one_{[0\le s\le\rho\w t]}\mE[\one_{[\rho<\infty]}|\mF_s]\Big]\md s
\\&=\int_0^\infty \mE\Big[\mE\big[\one_{[0\le s\le\rho\w t]}\one_{[\rho<\infty]}|\mF_s\big]\Big]\md s
\\&=\mE\Big[\int_0^{\rho\w t}\one_{[\rho<\infty]}\md s\Big]
\\&=\mE\Big[\one_{[\rho<\infty]}\rho\w t\Big]\;.
\end{align*}
Letting $t\to\infty$, yields 
\[
\tilde\phi_2(r)=\mE\Big[\one_{[\rho<\infty]}\rho\Big]=\phi_2(r)\;.
\]
Note that $\psi'_1(r)<0$ and $\phi_2'(r^*)\ge 0$. 
\\Due to $\phi_2(r^*)=0$ and $\phi_2(r)\ge 0$ it must hold $\phi_2'(r^*)\ge 0$. On the other hand, if $\tilde r:=\inf\{r\ge 0:\; \psi_1'(r)\ge 0\}\le r^*$ it must hold $\psi_1''(\tilde r)<0$ in order to ensure that $\psi_1$ solves the differential equation \eqref{eq3}. Since, it is a contradiction we can conclude $\psi_1'(r^*)<0$ for all $r\in[0,r^*]$. \hfill $\square$
%%%%%%%%%%%%%%%%%%%%%%%%%%%%%%%%%%%%%%%%%%%%%%%%%%%%%%%%%%%%%%%%%%%%%%%%%%%%%%%%%%%%
\subsubsection*{Proof of Lemma \ref{lem:phi1}}
First note that it obviously holds $\phi_1'<0$.
We can solve the differential equation \eqref{eq1} explicitly and obtain
\begin{equation}
\phi_1'(r)=-\frac{r^{-\frac{2b}{\delta^2}}e^{-\frac{2a}{\delta^2} r}}{\int_{ r^*}^\infty y^{-\frac{2b}{\delta^2}}e^{-\frac{2a}{\delta^2} y}\md y}\;.\label{phi1}
\end{equation}
Consider now 
\[
\frac{\phi_1'(r)}{\phi_1(r)}=-\frac{r^{-\frac{2b}{\delta^2}}e^{-\frac{2a}{\delta^2} r}}{\int_{r}^\infty y^{-\frac{2b}{\delta^2}}e^{-\frac{2a}{\delta^2} y}\md y}\;.
\]
Deriving $\frac{\phi_1'(r)}{\phi_1(r)}$ with respect to $r$ yields $\Big(\frac{\phi_1'(r)}{\phi_1(r)}\Big)'=-\frac{\phi_1'(r)}{\phi_1(r)}\Big\{-\frac{\phi_1''(r)}{\phi_1'(r)}+\frac{\phi_1'(r)}{\phi_1(r)}\Big\}$.
Using \eqref{phi1} we obtain
\[
\frac{\phi_1''(r)}{\phi_1'(r)}=-\frac{2a}{\delta^2}-\frac{2b}{\delta^2 r}\;,
\]
Let for simplicity $h(r):=\frac{\phi_1'(r)}{\phi_1(r)}$. Then
\begin{align*}
h'(r)&=-h(r)\Big\{\frac{2a}{\delta^2}+\frac{2b}{\delta^2 r}+h(r)\Big\}
=-h(r)\Big\{1+h(r)\Big\}-h(r)\Big\{\frac{2a}{\delta^2}+\frac{2b}{\delta^2 r}-1\Big\}\;.
\end{align*}
Note that $-h(r)>0$ and $\frac{2a}{\delta^2}+\frac{2b}{\delta^2 r}-1<0$ for $r>R$. That is, if for some $\hat r>R$ it holds $h(\hat r)=-1$ then $h'(\hat r)<0$, meaning that $h(\hat r)<-1$ for all $r>\hat r$. But this is a contradiction to
\begin{align*}
\lim\limits_{r\to\infty} h(r)=\lim\limits_{r\to\infty}\frac{\phi_1''(r)}{\phi_1'(r)}=-\frac{2a}{\delta^2}>-1\;.
\end{align*}
Therefore, it must hold $h(r)> -1$ for $r>R$. 
%\\Consider now the interval $[0,r^*]$. 
\\Because $\lim\limits_{r\to 0} h(r)=-\infty$, by the intermediate value theorem there must be an $r^*\in (0,R]$ such that $h(r^*)=-1$. \hfill $\square$
%%%%%%%%%%%%%%%%%%%%%%%%%%%%%%%%%%%%%%%%%%%%%%%%%%%%%%%%%%%%%%%%%%%%%%%%%%%%%%%%%%%%%%%%%%%%%%%%%%%%%%%%%%%%%%%%%%%%%%
\subsubsection*{Proof of Lemma \ref{sec0:lastexit}}
Note that $\lambda_y$ is not a stopping time, because $\{\lambda_y \le t\}\notin \mF_t$. 
\\Further, we know from lemma \ref{app:lem1} and Borodin \& Salminen \cite[p.\ 27]{bs} that $\lambda_y <\infty$ a.s. with
\[
\mP_r[0<\lambda_y\le t]=\int_0^t \frac{p(u;r,y)}{G_0(y,y)}\md u
\]
where $p(t;r,y)$ is the transition density of $\{r_t\}$ with respect to the speed measure $m$ with density $m'$ of $\{r_t\}$ (for the exact formula for $m$ and $m'$ confer \cite[p.\ 366]{ek} formula (1.4); the differential equation for $m'$ can be found in \cite[p.\ 18]{bs}) and $G_\alpha(r,y)$ is the Green function with 
\[
G_0(y,y)=\int_0^\infty p(t;y,y)\md t\;.
\]
Let $g(t;r,y)$ be the density of $\{r_t\}$ with respect to the Lebesgue measure. Then
\begin{align*}
g(t;y,y)=p(t;y,y)m'(y)\;.
\end{align*}
Therefore, using the above formula for $G_0(y,y)$:
\begin{align*}
\mP_y[0<\lambda_y\le t]=\int_0^t \frac{g(u;y,y)}{\int_0^\infty g(z;y,y)\md z}\md u\;.
\end{align*}
According to \eqref{density} the density $g(t;y,y)$ is given by
\begin{align*}
&g(t;y,y)=c(t)e^{-u(t,y)-v(t,y)}\Big(\frac{v(t,y)}{u(t,y)}\Big)^{q/2}I_{q}(2\sqrt{u(t,y)v(t,y)}),
%\\&c(t)=\frac{2a}{(e^{at}-1)\delta^2}
%\\&u(t,y)= \frac{2a}{(e^{at}-1)\delta^2} ye^{at},
%\\&v(t,y)=\frac{2a}{(e^{at}-1)\delta^2}y,
%\\&q=\frac{2b}{\delta^2}-1<0
\end{align*}
and $I_{q}$ is modified Bessel function of the first kind of order $q$. 
Using this explicit representation and $I_{q}(2c(t)ye^{at/2})=(c(t)e^{at/2})^qy^q\s_{m=0}^\infty \frac{(c(t)ye^{at/2})^{2m}}{m!\Gamma(m+q+1)}$, we obtain 
\begin{align*}
\frac{g(t;y,y)}{\int_0^\infty g(z;y,y)\md z}&=\frac{c(t)e^{-c(t)y(e^{at}+1)}e^{-aqt/2}I_{q}\big(2c(t)ye^{at/2}\big)}{\int_0^\infty c(z)e^{-c(z)y(e^{az}+1)}e^{-aqz/2}I_{q}\big(2c(z)ye^{az/2}\big)\md z}
%\\&=\frac{c(t)e^{-c(t)y(e^{at}+1)}e^{-aqt/2}(c(t)e^{at/2})^qy^q\s_{m=0}^\infty \frac{(c(t)ye^{at/2})^{2m}}{m!\Gamma(m+q+1)}}{\int_0^\infty c(z)e^{-c(z)y(e^{az}+1)}e^{-aqz/2}(c(z)e^{az/2})^qy^q \s_{m=0}^\infty \frac{(c(z)ye^{az/2})^{2m}}{m!\Gamma(m+q+1)}\md z}
\\&=\frac{c(t)e^{-c(t)y(e^{at}+1)}e^{-aqt/2}(c(t)e^{at/2})^q\s_{m=0}^\infty \frac{(c(t)ye^{at/2})^{2m}}{m!\Gamma(m+q+1)}}{\int_0^\infty c(z)e^{-c(z)y(e^{az}+1)}e^{-aqz/2}(c(z)e^{az/2})^q \s_{m=0}^\infty \frac{(c(z)ye^{az/2})^{2m}}{m!\Gamma(m+q+1)}\md z}
\end{align*}
By bounded convergence theorem, we can let $y$ go to zero and obtain
\begin{align*}
\frac{g(t;0,0)}{\int_0^\infty g(z;0,0)\md z}&=\frac{c(t)e^{-aqt/2}(c(t)e^{at/2})^q}{\int_0^\infty c(z)e^{-aqz/2}(c(z)e^{az/2})^q \md z}
= \frac{c(t)^{q+1}}{\int_0^\infty c(z)^{q+1} \md z}\;.
%\\&=\frac{c(t)^{q+1}{\int_0^\infty e^{aqz}c(z)^{q+1} \md z}\;.
\end{align*}
Note that indeed it holds by partial integration and using $-1<q<0$:
\begin{align*}
\Big(\frac{2a}{\delta^2}\Big)^{-q-1}\int_0^\infty c(z)^{q+1} \md z&=\int_0^\infty \big(e^{az}-1\big)^{-q-1}\md z
\\&=\frac 1{-aq}\int_0^\infty \big(e^{az}-1\big)^{-q}e^{-az}\md z
\\&=\frac 1{-aq}\int_0^\infty e^{aqz}\big(e^{az}-1\big)^{-q} e^{-(1+q)az}\md z
\\&\le\frac 1{-aq}\int_0^\infty e^{-(1+q)az}\md z = \frac{1}{-qa^2(1+q)}<\infty\;.
%
%\Big\{ \int_0^\infty e^{-az}e^{az(q+1)}\Big(\frac{2a}{\delta^2}\Big)^{q+1}\big(e^{az}-1\big)^{-q-1}\md z
%\\&\le  \int_0^\infty e^{-az}\Big(\frac{2a}{\delta^2}\Big)^{q+1}\md z\;.
%%&=-\Big(\frac{2a}{\delta^2}\Big)^{q+1}\frac{ \big(e^{ax}-1\big)^{-q}e^{aqx}}{qa}\Big|_0^\infty=\Big(\frac{2a}{\delta^2}\Big)^{q+1}\;.
\end{align*}
The last inequality follows because $e^{aqz}\big(e^{az}-1\big)^{-q}\le 1$ and $-1-q<0$. With similar arguments one obtains
\begin{align*}
\mE_0[\lambda_0]=\int_0^\infty t \frac{g(t;0,0)}{\int_0^\infty g(z;0,0)\md z}\md t= \int_0^\infty t \frac{\big(e^{at}-1\big)^{-q-1}}{\int_0^\infty \big(e^{az}-1\big)^{-q-1}\md z}\md t<\infty\;.
\end{align*}
\hfill $\square$
\section*{Acknowledgments}
\noindent
The research of the first author was funded by the Austrian Science Fund (FWF), Project number V 603-N35. Also, the first author would like to thank the University of Liverpool for support and cooperation.

\end{document}